\documentclass[onecolumn,tightenlines,superscriptaddress]{revtex4}       
\usepackage{natbib}
\usepackage{color}
\usepackage{graphicx}
\usepackage{mathptmx} 
\usepackage{amsmath}
\usepackage{amssymb}
\usepackage{amsthm}

\newtheorem{definition}{Definition}
\newtheorem{lemma}{Lemma}
\newtheorem{theorem}{Theorem}

\renewcommand{\cite}{\citep}
\begin{document}

\title{Cover-Encodings of Fitness Landscapes}

\author{Konstantin Klemm}
  \email{klemm@ifisc.uib-csic.es}
\affiliation{
  IFISC (CSIC-UIB), Campus Universitat de les Illes Balears, 
    E-07122 Palma de Mallorca, Spain
}

\author{Anita Mehta}
  \email{anita@bioinf.uni-leipzig.de}
\affiliation{Max Planck Institute for Mathematics in the Sciences, Inselstrasse 22, D-04103 Leipzig, Germany}
\author{Peter F. Stadler}
  \email{studla@bioinf.uni-leipzig.de}
\affiliation{Bioinformatics Group, Department of Computer Science and
  Interdisciplinary Center for Bioinformatics,
  University Leipzig, D-04107 Leipzig, Germany}
\affiliation{Max Planck Institute for Mathematics in the Sciences, Inselstrasse 22, D-04103 Leipzig, Germany}
\affiliation{Santa Fe Institute, Santa Fe, NM 87501, USA}

\begin{abstract}  
  The traditional way of tackling discrete optimization problems is by
  using local search on suitably defined cost or fitness landscapes. Such
  approaches are however limited by the slowing down that occurs when the
  local minima that are a feature of the typically rugged landscapes
  encountered arrest the progress of the search process. Another way of
  tackling optimization problems is by the use of heuristic approximations
  to estimate a global cost minimum. Here we present a combination of these
  two approaches by using \emph{cover-encoding maps} which map processes
  from a larger search space to subsets of the original search space. The
  key idea is to construct cover-encoding maps with the help of suitable
  heuristics that single out near-optimal solutions and result in
  landscapes on the larger search space that no longer exhibit trapping
  local minima.  We present cover-encoding maps for the problems of the
  traveling salesman, number partitioning, maximum matching and maximum
  clique; the practical feasibility of our method is demonstrated by
  simulations of adaptive walks on the corresponding encoded landscapes
  which find the global minima for these problems.

\end{abstract}

\maketitle

\section{Introduction} 
\label{sect:intro}

Fitness landscapes have proved to be a valuable concept in the
understanding of adaptation in evolutionary biology and beyond, by
visualizing the relationships between genotypes and effective reproductive
success \cite{Wright:32,Wright:67}.  This concept has been taken forward in
the field of evolutionary computation, where the performance of
optimization algorithms utilizing local search has often been described as
dynamics on a fitness landscape, see e.g.\ the book by
\citet{Engelbrecht:14}.

However, fitness functions alone do not determine the performances of local
search algorithms, which depend also on the structure of the search spaces
involved. These in turn are determined by two largely independent
ingredients: (1) the concrete representations of the configurations that
are to be optimized, referred to as encodings (2) Locality in the search
space, referred to as a move set.

For many well-studied combinatorial optimization problems and related
models from statistical physics (such as spin glasses), there is a natural
encoding. For instance, tours of a Travelling Salesperson Problem (TSP) are
naturally encoded as permutations of the cities concerned, while spin
configurations are encoded as strings over the alphabet $\{+,-\}$ with each
letter referring to a fixed spin variable. This natural encoding is usually
free of redundancy; any residual redundancies that occur usually arise from
simple symmetries of the problem which can easily be factored out. For
instance, TSP tours can start at any city so that they are invariant under
rotations, while many spin glass models are invariant under simultaneous
flipping of all spins. This natural or ``direct'' encoding is often
referred to as the \emph{phenotype space}, see e.g.\
\cite{Rothlauf:06,Neumann:10,Rothlauf:11,Borenstein:14}.

In biology, fitness is conceptually understood as a property (function) of
the genotype. It depends, however, on properties of higher-level structures
such as molecular structure, gene-regulatory networks, tissues, or organs,
i.e., on a phenotype. The relationship of genotype and fitness, therefore,
is a composition of a genotype-phenotype map and phenotype-dependent
fitness function. This decomposition has been studied extensively in
several distinct models systems, including RNA secondary structures,
\cite{Schuster:94a}, gene regulatory networks \cite{Ciliberti:07}, and
metabolic networks \cite{Dykhuizen:87,Flamm:10b}.  Here, we focus on the
abstract structure rather than the specifics of such models.

For a given encoding, irrespective of whether it is genotypic or
phenotypic, the performance of search crucially depends on the move
set. Here, we will consider only reversible, mutation-like moves. The
search space therefore is modeled as an undirected graph. More general
settings are discussed e.g.\ by \citet{Flamm:07a}. The cost function
assigned to a specific search space defines a \emph{fitness
  landscape}. Evolutionary algorithms can thus be viewed as dynamical
systems operating on landscapes, whose structure has, as a consequence,
been studied extensively in the field
\cite{Reidys:02a,Ostman:10,Engelbrecht:14}.

Continuing the analogy with biology in evolutionary computation, an
additional encoding $Y$, the so-called \emph{genotype space}, is often used
\cite{Rothlauf:03,Rothlauf:06}.  The genotype-phenotype relation is
determined by a map $\alpha:Y\to X\cup\{\varnothing\}$, where $\varnothing$
represents phenotypic configurations that do not occur in the original
problem, i.e., $y\in Y$ does not encode a feasible solution of the original
problem whenever $\alpha(y)=\varnothing$.  For example, a frequently used
genotypic encoding for TSP tours comprises binary strings for two cities
which represent their presence (\texttt{1}) or absence (\texttt{0}), for
each of the possible adjacencies \cite{TSP:06}. Most binary strings,
however, do \emph{not} correspond to TSP tours.

In practice, genotypic representations are usually chosen with a high
degree of redundancy to tackle optimization problems which often also
introduces neutrality, i.e., the appearance of adjacent configurations with
the same value of the cost function.  Detailed investigations of fitness
landscapes from molecular biology have shown that degrees of neutrality
\emph{can} facilitate optimization \cite{Schuster:94a,Reidys:02a} due to
the inclusion of extensive neutral paths which prevent trapping in
metastable states \cite{Schuster:94a,Fernandez:07,Yu:02,Banzhaf:06}. On the
other hand, ``synonymous encodings'' where genotypes mapping to the same
phenotype form tight clusters in the genotype space have been
advocated for the design of evolutionary algorithms
\cite{Rothlauf:06,Choi:08,Rothlauf:11}. Rather than having neutral
  paths connecting remote areas of the landscape, cost-equivalent
  configurations are locally clustered in synonymous encodings.

What is clear is that, empirically, the introduction of arbitrary
redundancy (by means of random Boolean network mapping) does not increase
the performance of mutation-based search \cite{Knowles:02}, suggesting that
the inclusion of redundancy should be suitably designed in order to
facilitate optimization.  One such approach was that of \citet{Klemm:12b},
which emphasized the utility of such inhomogeneous genotype-phenotype maps
via the idea that low-cost solutions could be enriched and optimization
made more efficient in genotype space if the size of the preimage
$|\alpha^{-1}(x)|$ of the phenotypes were anti-correlated with the cost
function $f(x)$ .  Of course, for such anti-correlations to be imposed,
$\alpha$ needs to become explicitly dependent on the cost function.

\section{Simplifying Landscape Structure by Encoding} 

Before delving into the technicalities, we present a conceptual outline of
the key ideas of this contribution.  Our starting point is the
twenty-year-old observation by \citet{Ruml:1996} that certain redundant
encodings of the Number-Partitioning Problem (NPP) allow simple, generic
optimization heuristics to find dramatically improved solutions. In
previous work \cite{Klemm:12b} we found that this approach was not limited
to the NPP, but that suitably chosen redundant encodings also improved the
performance of heuristics on several other combinatorial optimization
problems. In the present work, our objectives are to understand (a) why the
particular method used by \cite{Ruml:1996} works so well and (b) how it can
be generalized to essentially arbitrary combinatorial optimization problems
in a principled way.

We focus in this contribution on black-box-type optimization scenarios in
which the information on the cost function $f(x)$ is exclusively obtained
by evaluating it for specific configurations $x\in X$ in the search space
$X$. The sequence of these function evaluations is determined by the
optimization heuristic. Practical algorithms of this type propose
candidates $x\in X$ for evaluation based on past evaluation results. These
candidates are chosen locally in the vicinity of past successful candidates
with the help of rules that depend on the representation of $X$.  This
explicitly or implicitly defines a topological structure on $X$. For the
purpose of the present contribution we assume that the topology of the
search space $X$ is expressed by a notion of adjacency that is respected by
the search process.

Intuitively, the most important obstruction for local optimization
heuristics is the presence of a large number of local optima that trap the
search process. The aim of a redundant encoding, therefore, is to provide
an alternative representation $Y$ of the optimization problem that reduces
the number of local optima and makes it easier to find the globally optimal
solution. Formulated over $Y$, we would wish that
\begin{itemize} 
\item[(i)] neighborhoods in $Y$ are small enough to be searched in
  practice.
\item[(ii)] for every starting point there is a path to the global optimum 
  such that the cost function is decreasing, or at least non-increasing.
\end{itemize}
Condition (i) ensures that we still deal with local search heuristics,
while condition (ii) intuitively makes the landscape easy to search. Note
that condition (ii) does not make the optimization problem trivial, since
the heuristics still have to find an efficient path among possibly many
very long ones. Its real significance is that it rules out traps and
guarantees that simple downhill search will be successful eventually.

Is it possible at least in principle to construct such an encoding? The
prepartition encoding, which performed best for the NPP \cite{Ruml:1996},
provides an important hint. Each particular encoding $y\in Y$ corresponds
to a restricted version of the original optimization problem, i.e., it can
be seen as constraining the original search space $X$ to a subset
$\varphi(y)\subseteq X$. A deterministic approximation is then used to
solve the restricted problem on $\varphi(y)$. For every $y\in Y$ this
provides an upper bound on the cost function $\tilde f(y)$. Since the
encoding is chosen such that there is also a code $\hat y$ for the global
optimum $\hat x\in X$, i.e., $\varphi(\hat y)=\{\hat x\}$, the task now
becomes to find $\hat y$, which minimizes $\tilde f$ by construction. The
numerical results by \cite{Ruml:1996} suggest that this auxiliary problem
of minimizing the cost function of the encoding is much easier than the
original, despite the fact that the search space is much larger.
Below we show that this is case because (1) $\tilde f$ does a good job at
approximating the true solution $\tilde F(y)$ of the restricted
optimization problem on $\varphi(y)$ and (2) the perfect solutions
$\tilde F(y)$ give rise to landscapes with the desired properties mentioned
above.

This observation suggests a general construction for ``good'' landscape
encodings. The first step is the construction of a genotype space $Y$ and
an encoding scheme $\varphi$ that maps genotypes to restrictions of the
original problem rather than a particular phenotype $y$. This map has to
satisfy certain conditions discussed in detail in
Section~\ref{subsect:oracle} to be a good choice. The cost function then
enters by guiding, for every genotype $y\in Y$, a heuristic that solves the
restricted problem $\varphi(y)$.

Following the formal introduction of the general concepts, we construct
landscape encodings explicitly for several well-known examples. In
Section~\ref{sect:coarse} we focus on a particularly useful construction
that makes use of the fact that the restricted subproblems on $\varphi(y)$
can be seen as smaller instances of the same type of optimization problem,
or alternatively, as coarse-grained problems. We show in particular that
the NPP heuristic that motivated our approach is also of this type. In
Section~\ref{sect:num}, finally, we use numerical experiments to show that
the encoding scheme proposed here also works well in practice.

\section{A Theory of Encoding Representations} 

\subsection{Landscapes} 

Formally, an instance $(X,f)$ of a combinatorial optimization problem
consists of a finite set $X$ and a cost function $f:X\to\mathbb{R}$ on $X$.
The task of the combinatorial optimization problem $(X,f)$ is to find a
\emph{global minimum} $\hat x\in X$ so that $f(\hat x)\le f(x)$ for all
$x\in X$.

A \emph{landscape} $(X,\sim,f)$ consists of a finite set $X$ endowed with a
symmetric and irreflexive (adjacency) relation $\sim$ and a cost function
$f:X\to\mathbb{R}$. A point $x^*\in X$ is a strict local minimum in
$(X,\sim,f)$ if (i) $f(x^*)>f(\hat x)$ and (ii) there is no $x'\in X$ with
$f(x')<f(x^*)$ and an $f$-non-increasing path $x^*=x_0,x_2,\dots,x_k=x'$,
that is, $x_{i-1}\sim x_{i}$ and $f(x_{i-1})\ge f(x_i)$ holds for $0<i\le
k$. Note that a global minimum $\hat x$ is not a strict local minimum as
defined above.

For any $X'\subseteq X$ the restricted problem $(X',f_{|X'})$, where
$f_{|X'}(x)=f(x)$ for all $x\in X'$, consists in finding a $\hat x'\in X'$
so that $f(\hat x')\le f(x')$ for all $x'\in X'$. A restricted landscape 
$(X',\sim,f_{|X'})$ can be defined analogously.

\subsection{Oracle Function and Cover-Encoding Map} 
\label{subsect:oracle}

A key ingredient in our reasoning is to consider the global solutions
  of restricted optimization problems. This is formalized as follows:
\begin{definition} 
  The \emph{oracle function} $F:2^X\to\mathbb{R}$ of an optimization
  problem $(X,f)$ is 
  \begin{equation}
    F(X'):=\min_{x\in X'} f(x) 
  \end{equation}
  for all $X'\subseteq X$. We use the convention $F(\emptyset)=\infty$. 
\end{definition} 
We say that a subset $X'\subseteq X$ is \emph{good} if
$F(X')=F(X)$, i.e., if $X'$ contains a global optimum, and \emph{bad} if
$F(X')>F(X)$.  The oracle function is by definition monotonic in the
following sense:
\begin{equation} 
  X'' \subseteq X' \Longrightarrow F(X'')\ge F(X') 
\label{eq:monotone}
\end{equation}

We call $F$ an oracle function because in general there is no efficient
algorithm for computing it. In fact, if we had an efficient way to compute
$F$, we would already have solved the original optimization problem as
well. Nevertheless, it is a useful theoretical construct, as we shall see
below. First, it guides our construction of encodings of the original
optimization problem that have the potential of being easily solved, or at
least easier to solve. Second, it provides an inroad for constructing
practical heuristics \emph{provided} we can come up with a good
approximation for $F$.

We start by formalizing the idea of an encoding of a landscape.
\begin{definition} 
  A function $\varphi: Y\to 2^X$ is a \emph{cover-encoding map} for $X$ if it 
  satisfies 
\begin{description}
  \item[(Y1)] $\bigcup_{y\in Y} \varphi(y) = X$.
\end{description}
\end{definition}
Property (Y1) states that the collection of sets $\{\varphi(y)|y\in Y\}$ is a
set cover of $X$. The points $y\in Y$ can be thought as coding for a
particular element of this set cover. In the following, we will be
interested in cover-encoding maps that satisfy some or all of the following 
additional properties:
\begin{description}
\item[(Y0)] $\varphi(y)\ne\emptyset$. 
\item[(Y2)] For every $x\in X$ there is a $y\in Y$ such that
  $\varphi(y)=\{x\}$.
\item[(Y3)] There is $y\in Y$ such that $\varphi(y)=X$.
\end{description}
Note that both (Y2) and (Y3) imply (Y1). Axiom (Y0) excludes infeasible 
points in $Y$.  

It is not hard to see that cover-encoding maps always exist. In particular,
consider any subset
$Y\subseteq \mathfrak{P}_0(X) = 2^X\setminus\{\emptyset\}$, the set of
  non-empty subsets of $X$, such that (i) the singletons $\{x\}\in Y$ for
all $x\in X$ and (ii) $\{X\}\in Y$. Then the identity $\iota$ is obviously
a cover-encoding map that satisfies (Y0), (Y1), (Y2), and (Y3).
  
Now consider an optimization problem $(X,f)$ and let $\varphi:Y\to 2^X$ be a
cover-encoding map for $X$. We define $\tilde F: Y\to\mathbb{R}$ as the
composition of $\varphi$ with the oracle function of $(X,f)$, i.e., $\tilde
F(y) = F(\varphi(y))$. In the following we will be interested in the
relationship between the ``encoded'' optimization problem $(Y,\tilde F)$
and the original problem $(X,f)$.

If condition (Y2) is satisfied, there is $\hat y\in Y$ so that
$\varphi(\hat y)=\{\hat x\}$ for every global optimum of the original
problem. For most applications it is sufficient to find one global optimum,
hence we will consider the weaker condition:
\begin{description}
\item[(F0)] There is $\hat y\in Y$ so that (i) $|\varphi(\hat y)|=1$ and
  $F(\varphi(\hat y))=f(\hat x)$.
\end{description}
Condition (F0) simply states that there exists a code $y\in Y$ that
identifies a global optimum of the original problem $(X,f)$. This is
sufficient to consider $(X,f)$ and $(Y,\tilde F)$ as ``equivalent
optimization problems''.

The identity cover-encodings from $Y_{\max}:=\mathfrak{P}_0(X)$ and
$Y_{\min}:=\{ \{x\}|x\in X\} \cup \{ X\}$ are the extreme cases.
$Y_{\max}$ encodes all possible subproblems, while $Y_{\min}$ only encodes
the singletons, i.e., the evaluation of the cost function $f$ for every
$x\in X$, as well as the full optimization problem.

In this contribution, we are interested in search-based algorithms. Hence we
fix an adjacency relation $\sim$ on $Y$. For the landscape $(Y,\sim,\tilde
F)$ we consider the following three properties:
\begin{description} 
\item[(R1)] For every $y\in Y$ with $\tilde F(y)=\tilde F(\hat y)$ 
            there is a sequence 
            $y=y_0,y_1,\dots,y_k=\hat y$ such that $y_i\sim y_{i-1}$ for 
            $0<i\le k$ and $\tilde F(y_i)=\tilde F(\hat y)$. 
\item[(R2)] For every $y\in Y$ with $\tilde F(y)>\tilde F(\hat y)$ 
            there is a sequence 
            $y=y_0,y_1,\dots,y_k=\hat y$ such that $y_i\sim y_{i-1}$ for 
            $0<i\le k$, $\tilde F(y_k)=F(\hat y)$ and 
            $\tilde F(y_{i-1})\ge \tilde F(y_i)$.  
\item[(R3)] Every $y$ with $\varphi(y)\ne X$ has a neighbor $y'\sim y$ with
            $\varphi(y)\subset\varphi(y')$.
\end{description}
In plain words, (R1) ensures that all minimum-cost encodings are
  connected by paths staying at minimum cost. Under (R2), each
  configuration is the beginning of a path to a minimum-cost configuration,
  with the value of the cost function not increasing along the
  path. Property (R3) uses the fact that all configurations in $Y$ are
  subsets of $X$. It says that each configuration $y \in Y$ has a
  neighboring configuration properly containing $y$. It is worth noting
  that (R3) is independent of the oracle function $F$.

For identity cover-encodings introduced above, a natural definition of
adjacency is to set $y\sim y'$ and $y'\sim y$ whenever (i) $y\subseteq y'$,
(ii) $y\ne y'$, and (iii) if $y\subseteq y'' \subseteq y'$ then $y''=y$ or
$y''=y'$. That is, two sets are adjacent if they are adjacent in the Hasse
diagram for set inclusion. By construction, every $y\in Y$ is connected by
a sequence of adjacent sets to all singletons $\{x\}$ with $x\in y$
and to the full set $y=X$. Since $\varphi$ is the identity, (R3)
holds. Using that $y\subseteq y'$ implies $\tilde F(y)\ge \tilde F(y')$,
properties (R1) and (R2) also following immediately.

Taken together, the identity cover encodings demonstrate that cover
encodings and associated adjacencies satisfying (Y0) through to (Y3) as
well as (R1), (R2), and (R3) always exist.

\begin{lemma}
(R3) implies (R2) for any oracle function $F$.
\end{lemma}
\begin{proof}
  If $\varphi(y)=X$, then $\tilde F(y)=F(X)=f(\hat x)=\tilde F(\tilde y)$ by
  construction. Now consider an arbitrary starting point $y$. By (R3), there
  is a neighbor $y'\sim y$ such that $\varphi(y)\subset \varphi(y')$ and by
  equ.(\ref{eq:monotone}) we therefore have $\tilde F(y')\le \tilde F(y)$.
  Repeating the argument, we obtain a $\tilde F$-non-increasing sequence
  $y,y',y'',\dots,y^{(k)},\dots$ along which $\varphi$ is strictly
  increasing in each step. Since $X$ is finite, there is a finite $k$ so
  that $\varphi(y^{(k)})=X$ and thus $\tilde F(y^{(k)})=\tilde F(\hat y)$,
  i.e., (R2) is satisfied.
\end{proof}

The importance of conditions (R1) and (R2) stems from the following
observation:

\begin{theorem} \label{theo:nostrictloc}
  Suppose $(X,f)$, $\varphi:Y\to 2^X$, and the relation $\sim$ on $Y$ are
  chosen such that (Y1), (F0), (R1), and (R2) are satisfied. Then the
  landscape $(Y,\sim,\tilde F)$ has no strict local optimum.
\end{theorem} 
\begin{proof}
  Let $y\in Y$ be an arbitrary starting point. If $\tilde F(y)=\tilde
  F(\hat y)$ then $y$, by (R1), is not a local optimum but part of a
  connected neutral network that contains the global optimum $\hat y$. If
  $\tilde F(y)\ne\tilde F(\hat y)$ then $\tilde F(y)>\tilde F(\hat y)$. By
  (R2) there is a path with non-increasing values of $\tilde F$ that
  connects $y$ to a point $y'$ with $\tilde F(y')=\tilde F(\hat y)$. We
  already know that there is a path with constant values of $\tilde F$
  leading from $y'$ to the global optimum $\hat y$. Thus $y$ is connected
  by a $\tilde F$-non-increasing path to $\hat y$. Hence $y$ is, by
  definition, not a strict local optimum.
\end{proof}

In particular, the identity cover encodings satisfy the conditions of
Thm.~\ref{theo:nostrictloc} and thus their landscapes have no strict local
optima. There are, however, also very different general constructions with
this property. In the remainder of this section, we consider one example.

\begin{definition} 
  Let $(X,\sim_X,f)$ be an arbitrary landscape. Its \emph{square encoding}
  is the map $\varphi: X\times X \rightarrow 2^X$, $(\xi_,x)\mapsto \{\xi,x\}$
  for $(\xi,x) \in X\times X$. The neighbourhood relation $\sim_Y$ on
  $Y:=X\times X$ is given by
\begin{equation*}
(x_1,x_2) \sim_Y (\xi_1,\xi_2) \Leftrightarrow 
     (x_1=\xi_1 \wedge x_2 \sim_X \xi_2) \vee 
     (x_2=\xi_2 \wedge x_1 \sim_X \xi_1)  
\end{equation*}
\end{definition}
The graph $(Y,\sim_Y)$ is the Cartesian square of the graph $(X,\sim_X)$
\cite{Hammack:16}.  The idea behind this construction is to allow a local
search algorithm to keep track of the best solution so far in one variable
and use the other variable for exploration. Figure~\ref{fig:square} shows
an example.
\begin{lemma} 
  The landscape $X\times X,\sim_Y,\tilde F$ satisfies (Y0), (Y2), (F0),
  (R1), and (R2). In particular it has no strict local optima.
\end{lemma}
\begin{proof}
  Considering the properties of $\varphi$, (Y0) is obtained with
  $|\varphi(y)|>0$ for all $y \in Y$; (Y2) is fulfilled choosing $y=(x,x)$
  for any $x \in X$. This implies (Y0) so $\varphi$ is a cover-encoding
  map. We have (Y3) only in the trivial case $|X| \le 2$. Property (F0) is
  fulfilled with $\hat y = (\hat x,\hat x)$.

  For $y,y' \in Y$, we write $d_Y(y,y')$ for the standard graph distance,
  the length of a shortest path, between $y$ and $y'$; analogous notation
  for the distance $d_X$ on $(X,\sim_X)$. For $(x_1,x_2) \in Y$ and
  $(\xi_1,\xi_2) \in Y$ we have $d_Y((x_1,x_2),(\xi_1,\xi_2)) =
  d_X(x_1,\xi_1)+d_X(x_2,\xi_2)$.

  Now let $(x_1,x_2) = y \in Y \setminus \{(\hat x, \hat x)\}$.  Then $x_1
  \neq \hat x \neq x_2$. We assume, without loss of generality, $f(x_1) \ge
  f(x_2)$ (otherwise swap $x_1$ and $x_2$).  Because $(X,\sim_X)$ is
  connected, we find a neighbour $x' \sim_X x_1$ with $d_X(x',\hat
  x)=d_X(x_1,\hat x)-1$. With $y' = (x',x_2)$, we have $\tilde F(y') = \min
  \{ f(x'), f(x_2) \} \le f(x_2) = \tilde F(y)$ and $d_Y(y',\hat y) =
  d_Y(y,\hat y) -1$. For each element $y \in Y$ we thus find a $y' \in Y$
  that (i) is strictly closer to $\hat y$ than $y$ is; and (ii) does not
  evaluate at higher value than $y$ under $\tilde F$. Using the argument
  inductively at most $d_Y(y,\hat y)$ times, the desired sequences in (R1)
  and (R2) are constructed. Therefore properties (R1) and (R2) are
  fulfilled by $(Y,\sim_Y,\tilde F)$.  Theorem~\ref{theo:nostrictloc} now
  implies that there are no strict local minima.
\end{proof}

\begin{figure}
  \begin{center}
  (a)  \includegraphics[width=0.38\textwidth]{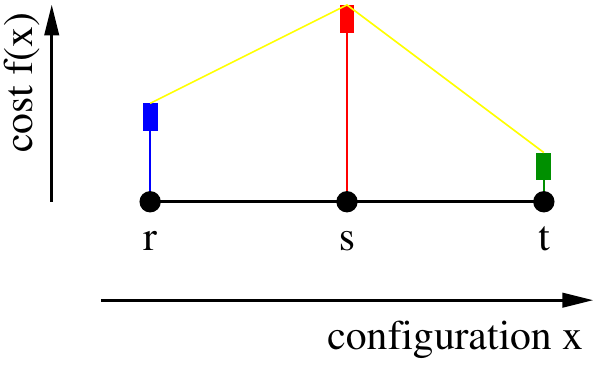} \hspace*{5mm}
  (b)  \includegraphics[width=0.50\textwidth]{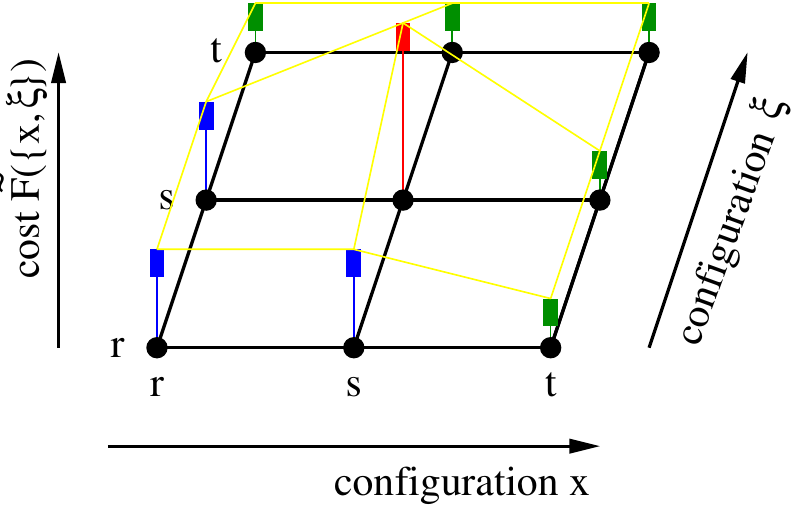}
  \end{center}
  \caption{Illustration of the square encoding. (a) Original landscape
      $(X,\sim,f)$ with configurations $X=\{r,s,t\}$. The three
      configurations form a path under the adjacency relation $\sim$. The
      cost function $f$ renders $t$ the unique global minimum, $r$ a strict
      local minimum. Thus $t$ is not reachable from $r$ by a non-increasing
      path.  (b) Landscape resulting from square encoding of the landscape
      in (a). Here each configuration is a tuple of configurations of the
      original landscape, $(x,\xi) \in X \times X$. The cost function is
      $\tilde F((x,\xi))= \min\{f(x),f(\xi)\}$. On this landscape, a
      minimal cost configuration is reachable from all configurations by a
      non-increasing path.}
  \label{fig:square}
\end{figure}

\subsection{Adaptive Walks} 

An adaptive walk on a fitness landscape $(Y,\sim_Y,\tilde F)$ is a Markov
chain on the state space $Y$ with transition probabilities
$\pi_{y \rightarrow z} = 1/d_y$ for $y \sim_Y z$ and
$\tilde F(z) \le \tilde F(y)$. Otherwise $\pi_{y \rightarrow z} = 0$,
except for $y=z$ where $\pi_{y \rightarrow y}$ is obtained by normalization
of probability.  The degree $d_y$ of state $y$ is the number of neighbours
$|\{ z \in Y : z \sim_Y y \}|$. Formulated as a stochastic search
  algorithm, a neighbour $z$ of the current (time $t$) configuration $y$ is
  drawn uniformly at random. If $\tilde F(z) \le \tilde F(y)$, the walk
  proceeds to configuration $z$ at time $t+1$; otherwise it remains at
  configuration $y$.

Call $\hat Y$ the set of global minima of the landscape
$(Y,\sim_Y, \tilde F)$. Assume that this landscape does not have a strict
local minimum.  Then each realization of an adaptive walk eventually
  hits a global minimum.  Due to the absence of strict local minima, the
  adaptive walk is trapped only at global minima. Each invariant measure of
  the adaptive walk therefore evaluates to zero on all configurations with
  non-minimum cost. Property (R2) clearly is a necessary condition for an
optimization problem to be solvable by adaptive walks alone. The conditions
of Theorem~\ref{theo:nostrictloc} are already sufficient as it excludes
strict local optima.

\subsection{Examples of Cover-Encoding Maps}

Let us now turn to constructing some problem specific examples of
cover-encoding maps. We will then use some of these examples to show that
some cover-encoding maps are useful to construct good heuristic search
algorithms for several well-studied combinatorial optimization problems.

\subsubsection{Prepartition Encoding for the NPP} 

An NPP instance is described by a list $(a_1,\dots,a_n)$ of numbers. We
write $[n]:=\{1,\dots,n\}$ for the index set. We have to divide these $n$
numbers into two subsets with as equal a sum as possible. In other words,
we assign to each index $i$ a variable $x_i\in\{-1,+1\}$ so that
\begin{equation}
  f(x) = \left| \sum_{i=1}^n x_i a_i  \right|  \to \min!
\end{equation} 
see e.g.\ \cite{Mertens:06} for a review. The set $X$ consists of all
  strings of $-1$ and $+1$ of length $n$, the set $Y$ consists of all
  functions $[n]\to[n]$. The so-called \emph{prepartitioning} encoding
\cite{Ruml:1996} of the NPP can be written in the following way: Each
  function $y:[n]\to[n]$ defines the partition
$\Pi_y:=\{ y^{-1}(k)|1\le k\le n\}$ whose classes are the indices of the
input numbers that are assigned the same value of $y$. As usual we write
$[i]_{\Pi_y}$ for the class $y^{-1}(k)$ that contains index $i$.  For given
$\Pi_y$ we now insist that the signs $x_i=x_j$ whenever $y(i)=y(j)$. This
amounts to the restricted set of configurations
\begin{equation}
  \varphi(y) = \{ x\in X| x_i=x_j \textrm{ if } j\in [i]_{\Pi_y} \}.
\end{equation}

One easily checks that $\varphi(y)=X$ whenever $y$ is a bijection, i.e., (Y3)
is satisfied. Furthermore, the subset $Y^*=\{ y\in Y \big| |y([n])|=2\}$
corresponds exactly to the assignments of positive and negative signs:
Writing $y([n])=\{p,q\}$ simply set $x_1=+1$ if $y(i)=p$ and $x_1=-1$ if
$y(i)=q$.  (More precisely, the choice of $x_1=+1$ or $x_1=-1$ is
arbitrary; the symmetry can, however, easily be removed e.g.\ fixing
$x_1=+1$ once and for all.) Conversely, every assignment of signs has a
representation as a bipartition in $Y^*$. Thus (Y2) is satisfied.

The most natural choice of an adjacency $\sim$ on $Y$ is to define $y\sim
y'$ if and only if $y(i)\ne y'(i)$ for exactly one $i\in[n]$. Unless $y$ is
a bijection, there is at least one unused value $k\in [n]\setminus y([n])$
and at least one pair $j',j''\in[n]$ with $y(j')=y(j'')$. The neighbor $y'$
of $y$ with $y'(i)=y(i)$ for $i\ne j''$ and $y'(j'')=k$ corresponds to
refinement of the partition $\Pi_y$ because
$[j']_{\Pi_y'}=[j']_{\Pi_y}\setminus\{j''\}$, $[j'']_{\Pi_y'}=\{j''\}$, and
all other classes of $\Pi_y'$ and $\Pi_y'$ are the same. Thus $(Y,\sim)$
satisfies (R3). 

An optimal solution $\hat x$ of the NPP $(X,f)$ is a partition $\hat\Omega$
of $[n]$ into exactly two classes $Q_+$ and $Q_-$ so that $x_i=+1$ for
$i\in Q_+$ and $x_i=-1$ for $i\in Q_-$. A code $y\in Y$ is good if there is a
configuration in $\varphi(y)$ in which the signs can be assigned in exactly
this manner, i.e., if $\Pi_y$ is a refinement of $\hat\Omega$. Conversely,
$\varphi(y)$ is good only if it is a refinement of a bipartition $\Omega$
that represent a global minimum. Generically $\hat\Omega$ is unique. Now
consider two classes $Q_1$ and $Q_2$ in $\Pi_y$ that are contained in the
small class of $\Omega$, i.e., $Q_1,Q_2\subset\Omega$. Reassigning one
element at a time from $Q_2$ to $Q_1$ thus corresponds to a sequence of
codes $y=y_1,y_2,\dots y_{|Q_2|}$ all of which are encode refinements
$\Omega$. Furthermore, $y_{|Q_2|}$ is one class less than $y$. Repeating
this step at most $n-2$ times eventually results in $\Omega$. Intermediate
codes $y_i$ and $y_{i-1}$ are adjacent by construction and satisfy $\tilde
F(y_i)=\tilde F(\hat y)$, i.e, condition (R1) is satisfied. Thus we
conclude that the ``oracle landscape'' $(Y,\sim,\tilde F)$ has no strict
local minima.
 
\subsubsection{Prepartition Encoding for the TSP} 

The cost function of TSP \cite{Gutin:07} is
\begin{equation}
f(\pi) = \sum_{i=1}^n d_{\pi(i),\pi(i+1)}
\end{equation}
where $\pi\in X$ is a bijection $\pi:[n]\to C$ from the index set $[n]$ to
a set of cities $C$.  The index $i$ specifies the position along the
tour. For a city $c$, therefore, $\pi^{-1}(c)$ is its position along the
tour. The problem is parametrized by distances $d:C\times C\to\mathbb{R}$
that satisfy $d(c,c)=0$ for all $c\in C$ but in general are neither
symmetric nor do they satisfy the triangle inequality.

\citet{Klemm:12b} introduced the following version of a prepartition
encoding. Here, an arbitrary function $y:C\to[n]$ is used to restrict the
possible orderings of the cities along the tour as follows: For all cities
$c,d\in C$, the condition $y(c)<y(d)$ implies
$\pi^{-1}(c)<\pi^{-1}(d)$. Again this defines a subset $X_y$ of the search
space $X$ of each $y$. We use the same definition of adjacency in
$Y$. Here, constant functions $y$ impose no restrictions on $\pi$, i.e,
$\varphi(y)=X$ whenever $y(c)=y(d)$ for all $c,d\in C$. On the other
hand, if $y$ is bijective then $X_y$ consists only of a single tour since
in this case $y(c)=\pi^{-1}(c)$ for all $c\in C$, i.e., $\pi=y^{-1}$. Thus
(Y2) and (Y3) are satisfied.

To address properties (R2) and (R1), we first observe that given an encoding
$y$, we can always move one city $c$ with $y(c)=k$ to one of the classes
defined by $y$ with an adjacent value $k'$. More precisely, suppose $k'$ is
such that (a) there is a city $d$ so that $y(d)=k'$ and b) there are no cities
$e$ with $y(e)=k''$, for any $k''$ between $k$ and $k'$. If $k'>k$, the city
which we can move is the one with $y(c)=k$ that appears last in the optimal
tour $\omega\in \varphi(y)$; similarly, if $k'<k$, we can move the city $c$
with $y(c)=k$ that appears first in the optimal tour $\omega\in
\varphi(y)$. In the first case, we can set $k < y'(c)\le k'$, while in the second
case, we can choose $k' \le y'(c) < k'$. By construction
$\omega\in\varphi(y')$, and therefore $\tilde F(y')\le \tilde F(y)$. It is
also clear from the construction that the step from $y$ to $y'$ can always
be chosen so that the number of classes $|y^{-1}([n])|$ remains constant,
increases by one $|y^{-1}([n])|$, or decreases by one -- unless we already
have $|y^{-1}([n])|=n$, in which case only a decrease is possible, or we
have $|y^{-1}([n])|=1$, in which case only an increase is possible. Thus we
can always find a path along which $\tilde F(y')$ does not increase
and along which $|y^{-1}([n])|$ is non-increasing or non-decreasing,
respectively. Note the moves keeping $|y^{-1}([n])|$ constant might be necessary to
 move the values $y(c)$ stepwise around in $[n]$ to have enough ``space''
to break up individual classes of $y^{-1}$, so that its members in the end
have consecutive values of $y$. It is not hard to convince oneself that
this is always possible. As a consequence, we can always connect any $y$ to
a code with a single class (for which $\varphi(y)=X$). For two adjacent
classes, we simply join, one-by-one, the cities of the smaller class to the
larger one.  Furthermore, the single-class code can be broken by pulling a
city at a time so that (R1) also holds. Note that (R3) is not necessarily
satisfied, however.

In contrast to the previous example of the NPP, here the paths are much
more involved and often longer. We therefore conjecture that the
prepartition encoding is less efficient for the TSP than for the NPP.

\subsubsection{Spanning Forest Encoding for the NPP} 
\label{sect:NPP-SF}

A very different encoding for the NPP can be constructed as follows. Denote
by $Y$ the set of all spanning forests of the complete graph $K_n$. For a
detailed discussion of the combinatorics of spanning forests we refer to
\cite{Teranishi:05}. For each forest $y\in Y$ denote by $y_a$ one of
  its connected components. Since $y_a$ is a tree and thus bipartite, there
  is a uniquely defined bipartition $(V_{y_a}^+,V_{y_a}^-)$ of its vertex
set. We assign $q_i=+1$ for $i\in V_{y_a}^+$ and $q_i=-1$ for
$i\in V_{y_a}^-$ to the other.
\begin{equation} 
  \varphi(y) = \{ x| x_i = p_aq_i,\, i\in V_{y_a},\, p_a=\pm 1 \}
\end{equation}
Suppose the spanning forest $y$ has $k$ components. Then the sign pattern
on each component $y_a$ is uniquely defined by fixing independently the
sign of the lexicographically smallest $i\in V_{y_a}$. Thus $\varphi(y)$
consists of exactly $2^k$ distinct configurations. It follows that
$\varphi(y)=X$ if $y$ contains no edges. Denoting the complement of
  $x$ by $\bar x$, we have $\varphi(y)=\{x,\bar x\}$ whenever $y$ is a
spanning tree.  Since $x$ and $\bar x$ represent the same solution of the
number partitioning problem, $\varphi$ satisfies (Y2) and (Y3).

(R3) holds since removing an edge from the spanning forest $y$ yields
another spanning forest $y'$ that imposes fewer restrictions and thus
corresponds to a larger subset of $X$. In general, write $y'\prec y$ if
$y'$ is a subforest of $y$. Then $\varphi(y)\subset\varphi(y')$. The
unconstrained search space corresponds to the spanning forest $y_0$
without edges. Conversely, every spanning tree $\hat t$ that defines the
bipartition of the globally minimal solution of the original NPP encodes
exactly this solution. Every sequence $\hat t = y_{n-1} \succ y_{n-2} \succ
\dots \succ y_1 \succ y_0$ of spanning forests obtained by successive edge
deletions from $\hat t$ connects $y_0$ and $\hat t$ and each $\varphi(y_i)$
also contains the global minimum encoded by $\hat t$. Thus (R1) holds.

\subsubsection{Subdivision Encoding for the TSP} 

An alternative encoding for the TSP uses a permutation $\psi:[n]\to C$ of
the set of $C$ cities and subdivision $\Pi$ of $[n]$ into consecutive
intervals. We specify $\Pi$ by the upper bound of the interval, i.e.,
$I_u:=\{k| i_{u-1}<k\le i_u\}$. Since the tours are circular, we set $i_0=i_m$
and as usual consider the order $<$ circular on $[n]$. Therefore
$I_1:=\{i_{m+1},\dots,i_n,1,\dots i_1\}$. An encoded configuration
$y:=(\psi,\Pi)$ fixes the order $\psi$ of cities $\psi(k)$ within each of
the index intervals $I_u$. The first city in interval $I_u$ is
$\psi(i_{u-1}+1)$, the last city is $\psi(i_{u})$.  Thus $\pi\in \varphi(y)$
if $\pi$ is obtained by permuting the intervals $I_u$ and following the
order given by $\psi$ within each interval, see Fig.~\ref{fig:TSPsubdiv}.

\begin{figure}
  \begin{center}
    \includegraphics[width=0.5\textwidth]{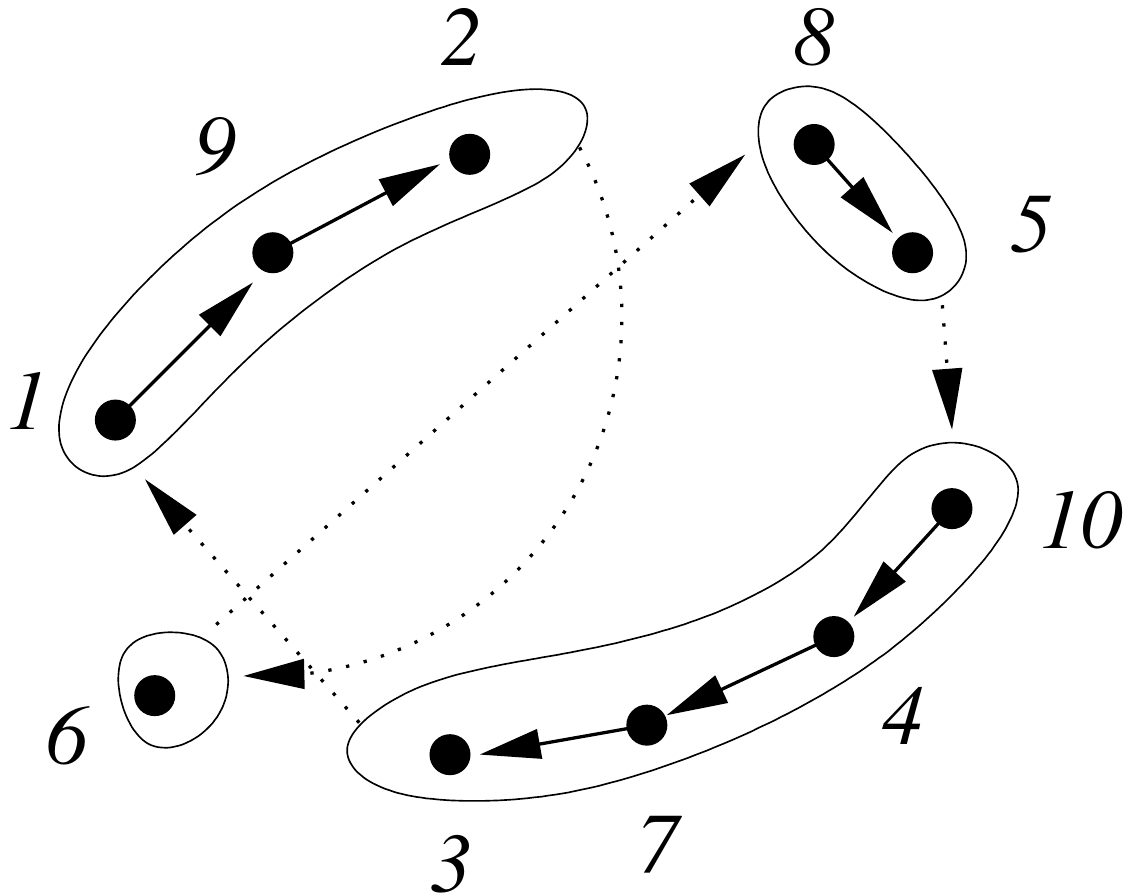}
  \end{center}
  \caption{Example for a subdivision of the TSP. The cities are subdivided
    into classes of a partition within which their order is fixed among all
    restricted tours (full arrows). The order in which the classes are
    traversed remains free (dotted arrows).} 
  \label{fig:TSPsubdiv}
\end{figure}

If $\Pi$ is the discrete partition, then we obviously have $\varphi(y)=X$,
while the indiscrete partition uniquely specifies the tour $\psi$. The
encoding therefore satisfies (F0), (Y0), (Y1), (Y2), and (Y3). Consider any
adjacency relation $\sim$ on $Y$ so that $y\sim y'$ if $\Pi'$ is obtained
by splitting a class (interval) into two or merging two intervals. Then (R3)
is clearly satisfied. 

In order to consider (R1) we specify the adjacency relation $\sim$
  more stringently. If $y\sim y'$ then either (i) $y$ is obtained from $y'$
  by splitting exactly one class of $y'$ into two non-empty parts or
  \emph{vice versa}, or (ii) $y$ and $y'$ exhibit the same partition of the
  cities, i.e., $\Pi=\Pi'$. In case (i), the ordering within each class in
  maintained. For the split interval
  $I_u'=[\psi'(i_{u-1}+1),\dots,\psi'(i_{u})]$, this means that an index
  $j\in[i_{u-1}+1,i_u-1]$ is chosen and the resulting intervals become
  $I_{u_1}=[\psi'(i_{u-1}+1),\dots,\psi'(j-1)]$ and
  $I_{u_2}=[\psi'(j),\dots,\psi'(i_{u})]$. The ordering between intervals
  (classes of $\Pi$) remains fixed. In case (ii), the partition and the
  ordering within the intervals both remain unchanged, but the ordering of
  the intervals (classes of $\Pi$) changes. For our purposes it is not
  important which types of permutations between intervals are allowed, as
  long as they form an ergodic set. Plausible choices are transpositions,
  canonical transpositions, reversals, or even all permutations.
  
Now consider an encoded configuration $\hat y$ with
$\hat x\in\varphi(\hat y)$. The intervals of specified $\hat y$ are partial
tours of the globally optimal solution. Moves on $Y$ can now be performed
so that a new encoding $y'$ is obtained in a stepwise fashion, that uses
the same intervals and brings two partial tours that are consecutive in
$\hat x$ into the desired order. During this stepwise change of $\psi$ the
encoded sets $\varphi(y)$ stay the same, and thus
$\varphi(y')=\varphi(\hat y)$. Now the two appropriate consecutive
intervals can be merged. This reduces $m$ by $1$ and makes $\varphi(y)$
smaller, but the globally optimal solution is still retained, i.e.,
$\hat x\in\varphi(y)$. The procedure can be repeated at most $m-1$ times to
reach the indiscrete partition, which fully specifies the globally optimal
tour. Thus (R1) holds for all choices of neighborhoods that allow
merging/splitting of adjacent intervals and an ergodic permutation of the
intervals.

\subsubsection{Sparse Subgraph Encoding for the Maximum Matching Problem}
\label{sec:sparse_mmx}

For a graph $G=(V,E)$, a matching is a subset $M \subseteq E$ of pairwise
disjoint edges, i.e.\ $(V,M)$ is a graph with a maximum degree  of at most 1.
Denoting by $X$ the set of matchings on $G$, the maximum matching problem
(MMP) $(X,f)$ has the cost function $f$ giving the number of unmatched
nodes
\begin{equation}
f(M) = \big|V \setminus \bigcup_{e \in M} e\big| 
\end{equation}
in a matching $M$. Thus the MMP asks for a subset of edges that cover as
many nodes as possible without having any node contained in more than one
edge \cite{Lovasz:86}.

Now consider an edge subset $S \subseteq E$. In the present context, we
call $S$ sparse if the graph $(V,S)$ has maximum degree 2, so each
connected component of $(V,S)$ is a cycle or path (including isolated nodes
as trivial paths). Denote by $Y$ the set of all sparse subsets of $E$.
Since a matching $M$ is also a sparse subset of $G$, we have
$X \subseteq Y$.

The cover-encoding map $\varphi:Y \rightarrow 2^X$ assigns each $S \in Y$ the
set of maximum matchings of the graph $(V,S)$.  Now with $S$ sparse, the
maximum matching problem on $(V,S)$ is trivially solved separately on each
connected component being a path or cycle. For a path of odd length $k$,
the maximum matching is unique with $(k+1)/2$ edges; a path or cycle of
even length $k$ has exactly two disjoint maximum matchings of cardinality
$k/2$. A cycle of odd length $k$ has exactly $k$ pairwise different maximal
matchings of cardinality $(k-1)/2$.

For each matching $x \in X$, we have $\varphi(x) = \{x\}$ so property (Y2)
holds. Properties (Y0) and (Y1) are fulfilled.  With the choice $\hat y =
\hat x$, (F0) is fulfilled. Property (Y3) holds if and only if $(G,E)$ is
sparse itself.

We consider sparse subsets $D$ and $D^\prime$ as adjacent, $D \sim
D^\prime$, if they differ at exactly one edge, $|(D\cup D^\prime) \setminus
(D\cap D^\prime)|=1$.

In order to demonstrate properties (R1) and (R2), let
$y \in Y \setminus \{ \hat y \}$. We show that there is $y' \sim_Y y$ with
$\tilde F(y') \le \tilde F(y)$ and
$|(y' \cup \hat y) \setminus (y' \cap \hat y)| \le | (y \cup \hat y)
\setminus (y' \cap \hat y)|$. Thus neighbour $y'$ is obtained from $y$
either by adding an edge contained in $\hat y$ or removing an edge not
contained in $\hat y$. If $y \supset \hat x$, find an edge
$e \in y \setminus \hat x$ and set $y'=y \setminus \{e\}$, and we are
done. Otherwise, since $y \neq \hat y$, there is an edge
$\{v,w\}=e \in \hat x \setminus y$. If $y \cup \{e\}=:z$ is sparse, we are
done using $y' = z$. Otherwise at least one of nodes $v$ and $w$ has degree
3 in the graph $(V,z)$; suppose node $v$ has degree 3. Find a maximum
matching $x \in \varphi(y)$. Since $v$ has degree 2 in the graph $(V,y)$,
there is an edge $e'\in y \setminus x$ incident in $v$. Set
$y'=y\setminus \{e'\}$. We easily confirm $\tilde F(y') \le \tilde F(y)$ in
each of the cases above. Sequences for properties (R1) and (R2) are
obtained by induction.

\subsubsection{String Encoding for the Maximum Clique Problem}
\label{sec:seq_clique}

For a graph $G=(V,E)$, a clique is a node subset $C \subseteq V$ 
inducing a fully connected subgraph, i.e.\ $\{v,w\} \in E$ for all
$v,w \in C$ with $v \neq w$. Denoting by $X$ the set of cliques of $G$,
the maximum clique problem (MCP) $(X,f)$ has the cost function $f$
giving the number of nodes
\begin{equation}
f(M) = |V \setminus C|
\end{equation}
outside a clique $M$ \cite{Bomze:99}.

For arbitrary $l \in \mathbb{N}$ and any string of not necessarily
distinct nodes
$(v_1,v_2,\dots,v_l) \in V^l$, we define the greedy clique
$\gamma_G(v_1,v_2,\dots,v_l) \subseteq V$ recursively by 
\begin{equation}
\gamma_G(v_1,v_2,\dots,v_l) = \left\{
\begin{array}{ll}
\gamma_G(v_1,v_2,\dots,v_{l-1}) \cup \{v_l\} & 
\textrm{if } \{v_i,v_l\} \in E \textrm{ for all } i \in [l-1]\\
\gamma_G(v_1,v_2,\dots,v_{l-1}) & \textrm{otherwise}
\end{array} \right.
\end{equation}
and $\gamma_G(\varnothing)=\emptyset$ for the empty string $\varnothing$.

We construct a cover-encoding map $\varphi$ based on strings of length
$|V|=:n$, so $Y=V^n$. For a string $y \in Y$, we denote the substring 
(suffix) from index $k$ to the end (index $n$) by $(y)_{k\dashv}$.
Now $\varphi$ maps a string $ y\in Y$ to maximal greedy
cliques over suffices of $y$,
\begin{equation}
\varphi(y)= \{ \gamma_G((y)_{k\dashv}) : k \in [n] \textrm{ and } 
\forall i \in [n] :
\gamma_G((y)_{k\dashv}) \not\subset \gamma_G((y)_{i\dashv})  \}~.
\end{equation}
So a clique $C$ is contained in $\varphi(y)$ if and only if $C$ is a greedy
clique from a suffix of $y$ and none of the other greedy cliques from $y$
properly contains $C$. This ensures that $\varphi$ produces all the
singletons, thus fulfilling property (Y2). We call $y$ pure if
$|\varphi(y)|=1$. A string $y \in Y$ is pure if and only if $\{ y_i : i \in
[n] \}$ is a clique of $G$. We define strings $y, y' \in Y$ to be adjacent,
in symbols $y \sim_Y y'$, if and only if there is a unique index $i \in
[n]$ with $y_i \neq y_i'$ (Hamming distance 1).

In order to prove properties (R1) and (R2), we first observe that there is
a non-increasing sequence of strings from any $y \in Y$ to a pure
$y^{(\mathrm p)} \in Y$ with $\varphi(y^{(\mathrm p)}) \subseteq
\varphi(y)$ and $\tilde F(y^{(\mathrm p)}) = \tilde F(y)$. The sequence is
obtained by finding a maximal $C \in \varphi(y)$. If $y$ is not pure, there
is $i \in [n]$ with $y_i \notin C$. The next string in the sequence can be
obtained by replacing the entry $y_i$ with an arbitrary element from $C$.

If $y, z \in Y$ are pure with $\varphi(y)=\varphi(z)=\{C\}$ and $|C|<n$,
there is a non-increasing sequence from $y$ to $z$. It may be constructed
by stepwise swapping operations. Since $|C|<n$, there is at least one
element in $C$ found at two distinct positions in $y$ so one of these can
be used as a temporary variable in the swap.

Now let $y, y' \in Y$ with $\tilde F(y') \le \tilde F(y)$. Find a maximal
clique $C \in \varphi(y)$ and a maximal clique $C' \in \varphi(y')$. We
construct a non-increasing sequence from $y$ to $y'$ by concatenating the
following sequences. First, a non-increasing sequence from $y$ to a pure
$y^{(\mathrm p)} \in Y$ with $\tilde F(y^{(\mathrm p)}) = \tilde
F(y)$. Second, a non-increasing sequence from $y^{(\mathrm p)}$ to a pure
$z \in Y$ with $\{ z_1,z_2, \dots, z_{|C|}\}=C$ and $\{ z_1,z_2, \dots,
z_{|C \setminus C'|}\}=C \setminus C'$, and arbitrary $z_{|C|+1},z_{|C|+2},
\dots, z_n \in C$. Third, a sequence from $z$ to a string $z'$ is obtained
by assigning, step by step, nodes in $C'\setminus C$ to entries from
$z_{|C|+1}$ to $z_n$. The sequence is non-increasing because each of its
strings generates $C$ under $\varphi$. On the other hand,
$\gamma_G((z')_{(|C \setminus C'|+1)\dashv}) = C'$ so $\tilde F (z') =
\tilde F(y')$. Now again by swap steps, we transform $z'$ into $y'$.

\section{Coarse-Graining} 
\label{sect:coarse}

Some of the restricted search spaces $\varphi(y)$ introduced above can also
be thought of as coarse-grainings of the original problem. In the following
subsections, we show this for the prepartition and spanning forest
encodings of the NPP, as well as for the TSP.

\subsection{Prepartition Encoding of the NPP} 

Consider the NPP instance with numbers $\{a_1,a_2,\dots,n\}$ and let
$\Pi=\{Q_1,\dots,Q_m\}$ be an arbitrary partition of $[n]$ with classes
(subsets) $Q_j$ so that $m\le n$. Of course, we can think of $\Pi$ as the
classes defined by the prepartition encoding, i.e.\
$\Pi = \{y^{-1}(k)| k\in[n]\}$.  Set $b_j=\sum_{i\in Q_j} a_i$. Then the
set of numbers $\{b_1,\dots,b_m\}$ defines an NPP on $m$ numbers. In terms
of a prepartition $y$ this amounts to $b_k = \sum_{i \in y^{-1}(k)} a_i$.
Note that if $m=n$, then $\Pi$ is the discrete partition in which every
class $Q_j$ contains only a single element, and hence
$\{a_1,\dots a_n\}=\{b_1,\dots,b_m\}$. In the general case the solutions of
the two NPPs are related to each other in the following way.  Denote the
variables for the smaller NPP by $x'_j\in\{+1,-1\}$ and write $f_a$ and
$f_b$ for the cost functions. Then, obviously
\begin{equation}
  f_a(x) = f_b(x') \textrm{ whenever } x_i=x'_j \textrm{ for all } i\in Q_j
\end{equation}
An optimal solution $\hat x$ of the larger problem $(X,f_a)$ corresponds to
a partition $\hat\Omega$ of $[n]$ into exactly two classes $Q_+$ and $Q_-$
so that $x_i=+1$ for $i\in Q_+$ and $x_i=-1$ for $i\in Q_-$. The
coarse-grained NPP $(X',f_b)$ has an optimal solution with the same cost if
(and in the generic case also only if) $Q_j\subseteq Q_+$ or
$Q_j\subseteq Q_-$ holds for all $j\in[m]$, i.e., if (and generically only
if) the coarse-graining partition $\Pi$ is a refinement of the partition
$\hat\Omega$ that encodes the globally optimal solution of the original
problem.

\subsection{Travelling Salesman Problems} 

Recall the subdivision encoding for the TSP and fix an encoding
$y=(\psi,\Pi)$.  The length of the partial tour inside the interval $I_u$
is
\begin{equation}
  \ell_u = \sum_{k=i_{u-1}+2}^{i_u} d_{\psi(k-1)\psi(k)}
\end{equation}
Furthermore, the road from interval $I_p$ to interval $I_q$ is the road 
from $\psi(i_p)$ to $\psi(i_{q-1}+1)$, i.e., 
\begin{equation} 
  \tilde d_{pq} = d_{\psi(i_p),\psi(i_{q-1}+1)}
\end{equation} 
Since a tour $\pi\in\varphi(y)$ is uniquely defined by a permutation 
$\xi:[m]\to[m]$ of the intervals, we have 
\begin{equation} 
  \ell(\pi) = \tilde\ell(\xi) + \sum_{u=1}^m \ell_u
\end{equation}
where $\tilde\ell(\xi)=\sum_{i} \tilde d_{\xi(i),\xi(i+1)}$ is the 
tour length of the TSP restricted to the connections between the fixed 
intervals. With a slight change one can also produce a TSP that retains the 
original values of the cost function. To this end we set 
\begin{equation} 
   d_{pq}' = d_{\psi(i_p),\psi(i_{q-1}+1)} + (\ell_p+\ell_q)/2
\end{equation} 
and $\ell'(\xi):= \sum_{i} \tilde d'_{\xi(i),\xi(i+1)}$. A short
computation verifies $\ell(\pi)=\ell'(\xi)$. 

Note that we naturally obtain an asymmetric TSP even if the original
problem was symmetric since now $d'_{pq}\ne d'_{qp}$ because in general we
will have $d_{\pi(i_p)\pi(i_{q-1}+1)}\ne d_{\pi(i_q)\pi(i_{p-1}+1)}$. 

\subsection{Spanning Forest Representation of the NPP} 

Let us now return to the NPP. Let $y$ be a spanning forest of $K_n$. 
For each connected component (tree) $t\dot\subseteq y$ let $V^+_t$ and 
$V^-_t$ be the corresponding bipartition of the vertex set of $t$. Define 
\begin{equation} 
  b_t  = \left|  \sum_{i\in V^+_t} a_i -  \sum_{i\in V^-_t} a_i \right| 
\end{equation} 
This defines an instance of the NPP with as many numbers $b_t$ as connected
components in $y$. A choice of sign $z_t\in\{+1,-1\}$ for $t$ implies a
particular choice of sign for each $a_i$, i.e., each configuration $z$ for
the NPP with numbers $\{b\}$ corresponds to a configuration $x$ of the
original problem with numbers $\{a\}$. Clearly, these coincide with the
configurations $\varphi(y)$ described in Sect.\ \ref{sect:NPP-SF}.

\subsection{Some remarks on coarse-grainings: analogies with the renormalization group?} 

It is tempting to speculate that the coarse-grainings we have observed in
the above are analogous to those observed in renormalization group theory,
well known for its use in analyzing spin glasses and related disordered
systems \cite{Rosten:12}. In our context it can be described as
follows. For a given type of problem, such as the NPP or the TSP, consider
the space $\mathfrak{X}$ of all possible instances of all sizes. A
particular instance (e.g.\ the NPP with $n$ numbers
$a=\{a_1,a_2,\dots,a_n\}$) is a point $\mathbf{x}\in\mathfrak{X}$. Now we
define a set $\mathcal{R}$ of maps $r:\mathfrak{X}\to \mathfrak{X}$ that
map larger instances to strictly smaller ones. Of interest in this context
are in particular those maps $r$ that (approximately) preserve salient
properties. Since $r(\mathbf{x})$ is a smaller instance than $\mathbf{x}$,
the map $r$ is not invertible. The maps in $\mathcal{R}$ can of course be
composed, and thus form a semi-group which is known as the
\emph{renormalization group} \cite{Wilson:74,Wilson:71}.  Of course, while
renormalization groups in statistical physics are used to analyse the
typical behavior of large systems near criticality, our focus in the
present optimization context is on particular instances of systems that are
typically large. This does not yet rule out an analogy, assuming that
something like an ergodic hypothesis applies, where the behaviour of
typical instances is indeed that of the average. Thus, starting from
$\mathbf{x}=(X,f)$, or more precisely, an encoding $y$ so that
$\varphi(y)=\mathbf{x}$, we can think of adjacent encodings $y'\sim y$ with
$|\varphi(y')|<|\varphi(y)|$ as ``renormalized'' versions of
$\varphi(y)$. A path in $(Y,\sim)$ leading from $\mathbf{x}$ to the trivial
instance thus can be seen as the iteration of progressively renormalized
samples.

A positive example of this analogy could be that of the spanning forest
encoding of the NPP with real-space renormalization schemes for Ising
spins: an example of an $\mathcal{R}$ could be a so-called block spin
transformation \cite{Kadanoff:66}, where suitable averages are taken over
small local subsets of spins, which are then progressively scaled up to
larger system sizes to explore their critical behaviour.  Only certain
block variables will work for such schemes, depending on the underlying
symmetries of the problem, just as, in the earlier subsection, only the
sums of numbers $a_i$ preserve the optimal solutions.  Such simple
real-space scalings, do not, however, always exist for our optimization
schemes: the prepartition encoding of the TSP, for example, cannot be
rephrased as a coarse-grained (i.e., reduced-size) TSP. To see this, simply
observe that the evaluation of a tour in the restricted model still
requires an optimization over multiple incoming and outgoing connections
(roads) for every city, i.e., the information of inter-city distances
cannot by collapsed in any way upon the transition from a larger (less
restricted) to a smaller (more restricted) problem. This does not, however
rule out the possibility of, say, a renormalisation-type scaling in some
sort of generalised Fourier space. In the case of landscapes on
  permutation spaces, the characters of the symmetric group provide a
  suitable Fourier-like basis \cite{Rockmore:02a}, which seem to be
  applicable to TSP and certain assignment problems. These and other
possibilities are currently being explored, since it seems that deep
similarities may underlie relatively superficial differences in the nature
of the transformations involved in renormalization groups and the
optimization-facilitating encodings that are the subject of this paper.

\section{Heuristic Optimization over $Y$}
\label{sect:num}

\subsection{General Considerations} 

So far, we were only concerned with the abstract structure of
cover-encoding maps $\varphi:Y\to 2^X$ and the adjacencies $\sim$ in their
encodings $Y$. On this theoretical basis, we can now construct a
search-based \emph{optimization heuristic} that generalizes the approaches
in \cite{Ruml:1996} and our earlier work \cite{Klemm:12b}. The idea is very
simple: If we have an accurate and efficiently computable heuristic, we can
quickly obtain good upper bounds $\alpha_f(y)\ge \tilde F(y)$ for each of
the restricted problems $(\varphi(y),f)$. The properties (R1) and (R2)
guarantee the existence of non-increasing paths from an arbitrary initial
encoding $y_0$ down to a final encoding $\hat y$. Steps to adjacent
encodings that decrease $\alpha_f$ therefore will have a bias toward the
optimal solution of the original problem.

The fact that we have to rely on the quality of the estimate
$\alpha_f(y)\approx \tilde F(y)$ also suggests that it should be more
efficient to restart the search often rather than try to overcome barriers
of local minima in the landscape $(Y,\alpha_f)$. In the examples above
local minima in $(Y,\alpha_f)$ can, as we have proved, appear only due to
insufficient accuracy of the heuristic solutions $\alpha_f(y)$ for some
encodings. 

The discussion above also implies guidelines for the construction of
encodings:
\begin{enumerate} 
  \item The cover-encoding map $\varphi:Y\to X$ should be of a form that
    guarantees that $(Y,\sim,\tilde F)$ has no local optima, i.e., the 
    properties (R1), (R2), (Y1), and (Y2) should hold.
  \item The paths in $(Y,\sim)$ connecting large sets $\varphi(y)$ to smaller
    ones should not contain many steps along which the sets do not shrink.
    For instance, while the prepartition encoding for the NPP always has a
    strictly coarse-grained neighbor, this is not the case for the
    prepartition encoding for the TSP. We therefore suspect that other
    encodings for the TSP will work better in general. 
  \item The heuristic producing $\alpha_f(y)$ needs to be efficient,
    ideally not much slower than the function evaluations for the initial 
    cost function $f$. 
\end{enumerate} 

In order to demonstrate that the theory developed above may also have
practical implications we probe instances of encoded landscapes by adaptive
walks. To simulate a realization of an adaptive walk, we first generate an
initial state $y(0)$ by a procedure specific for the given landscape. At
each time step $t$, we uniformly draw a neighbor z of state $y(t)$ and set
$y(t+1)=z$ if $\tilde F(z) \le \tilde F(y(t+1))$, $y(t+1)=y(t)$ otherwise.

We select the MMP and the MCP as examples because (1) oracle functions and
encodings can constructed that guarantee the absence of strict local
minima; and (2) there is a simple and efficient algorithm for exact
computation of $\tilde F(y)$ for each $y \in Y$. So we do not require
heuristics. We leave the combination of cover-encoding maps with
non-trivial heuristics for a future manuscript.

\subsection{Maximum Matching Problems}

\begin{figure}[t]
\begin{center}
  \includegraphics[width=0.75\textwidth]{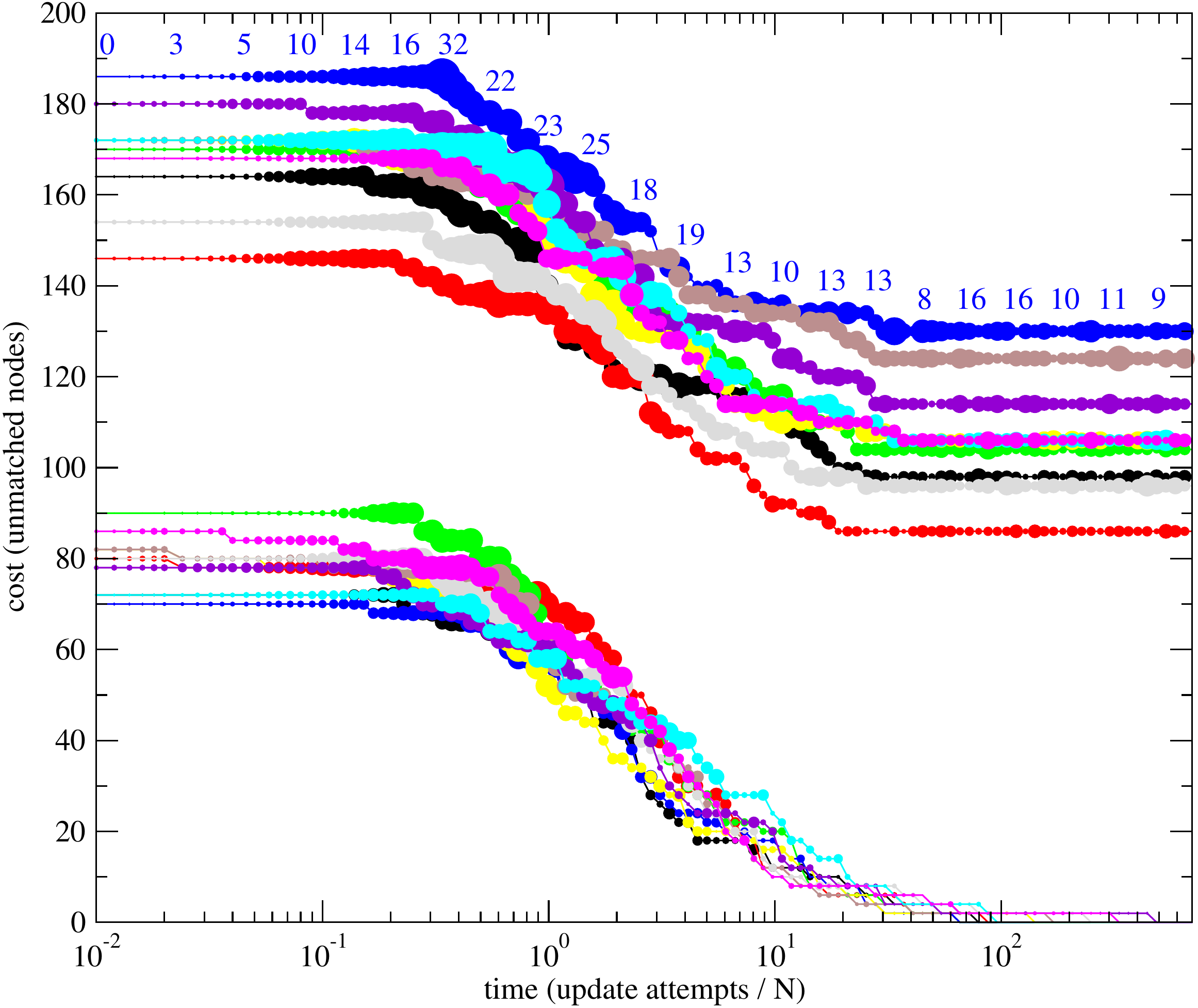}
\end{center}
\caption{Time evolution of cost in adaptive walks on the landscape of
  matchings encoded by sparse subgraphs. Radius of symbols is proportional
  to the number of degrees of freedom (paths of even length $\neq 0$ and
  cycles of odd length) in the encoded state.  Upper set of curves: 10
  realizations, each on an independently generated ER random graph on 500
  nodes with edge probability $p=2/(N-1)$, i.e.\ average degree 2.  Lower
  set of curves: 10 realizations on graphs (500 nodes) with perfect
  matching planted first, then adding each of the remaining possible edges
  with $p=1/(N-2)$, resulting in average degree 2. Each adaptive walk is
  initialized by a random maximal matching $L(0)$. Departing from the empty
  set, $L(0)$ is generated by considering the edges of the graph $G$ in the
  order of a random uniform permutation and adding an edge to $L(0)$ if the
  result remains a matching.  }
  \label{fig:matchingrg_0}
\end{figure}

Figure~\ref{fig:matchingrg_0} shows the time evolution of cost in adaptive
walks on the encoded landscapes of matchings encoded by sparse graphs,
where the figure caption contains details on the instances and the
definitions are to be found in section~\ref{sec:sparse_mmx}. Note the
logarithmic time axis in the plot.

Both on purely random graphs and on those with a planted perfect matching,
a solution of globally minimal cost is found. In addition to reaching a
minimum cost solution, we observe another interesting feature of the
dynamics. The sizes of symbols (and annotated values in the uppermost
curve) indicate the number of degrees of freedom $\delta = \log_2
|\varphi(y(t))|$ of the solution $y(t)$ at time $t$. This is the number of the
connected components in the sparse graph,  with two distinct maximum
matchings. Departing from a singleton state ($\delta=0$), the number of
degrees of freedom first increases and then decreases during the descent of
cost. So the optimization happens as a walk through states $y \in Y$ with
large cardinality $|\varphi(y)|$ of the encoded set. Furthermore as a
particular feature of this encoded landscape, the optimization dynamics
eventually returns to low $\delta$, having $|\varphi(y(t))|=1$ with a single
optimal solution selected at large time $t$.

\begin{figure}[t]
\begin{center}
  \includegraphics[width=0.75\textwidth]{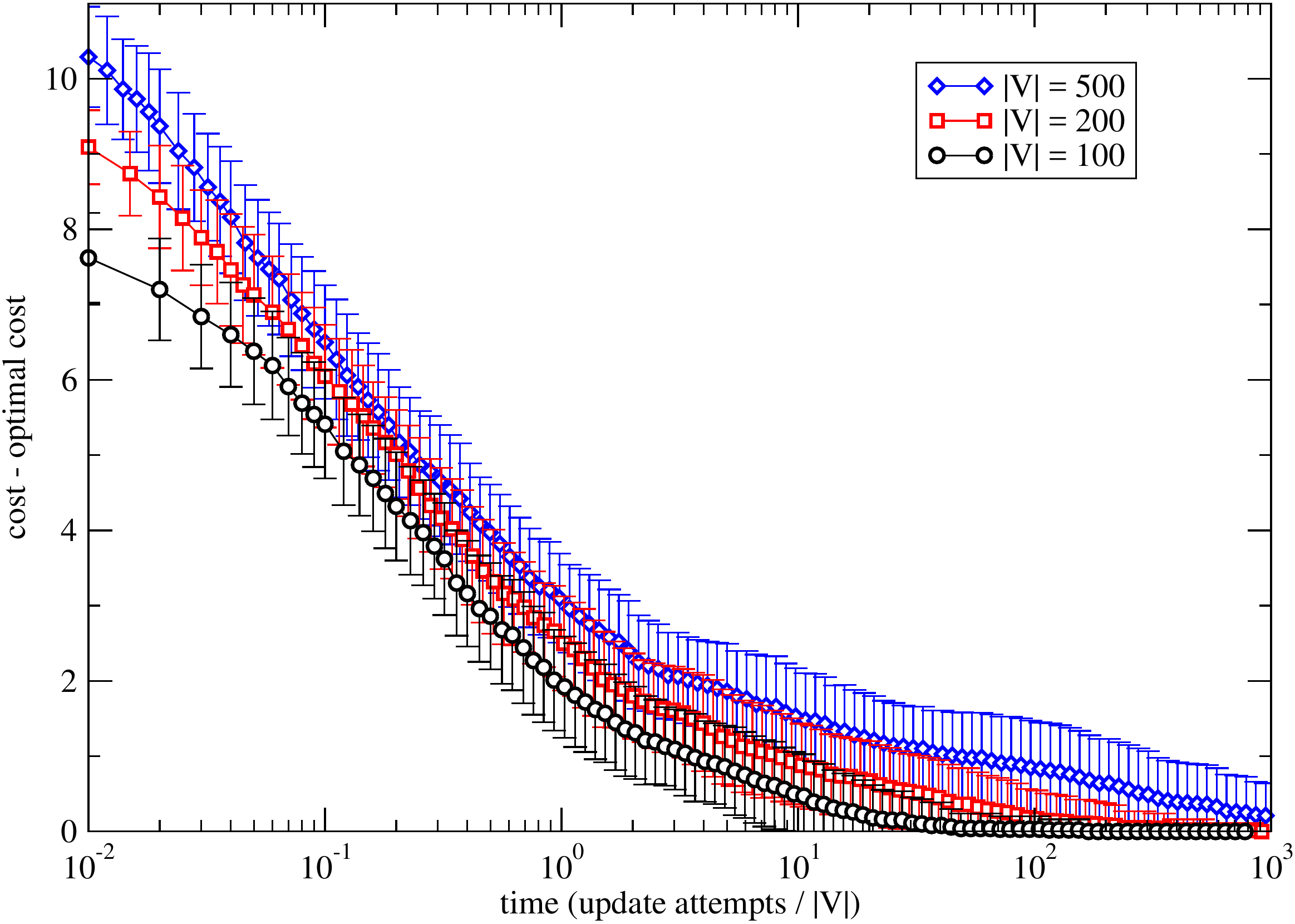}
\end{center}
\caption{Time evolution of cost in adaptive walks
  on the landscape of cliques encoded by node sequences. For each graph
  size $|V|$, 100 random graph instances with parameter $p=1/2$ are
  generated independently. For each instance, an adaptive walk on the
  encoded landscape is performed with starting state $(1,1,\dots,1)$.
  Plotted values are differences between $\tilde F(y(t))$ of the state
  $y(t)$ held by the adaptive walk at time $t$ and the optimal cost $F(\hat
  x)$, averaged over the 100 instances. Length of error bars is the
  standard deviation over these instances.  The exact $F(\hat x)$ is
  computed with a branch-and-bound algorithm \cite{Ostergard:02a}.  }
\label{fig:cliquergc_0}
\end{figure}

\subsection{Maximum Clique Problems}

Figure~\ref{fig:cliquergc_0} shows the time evolution of the cost of
adaptive walks on the encoded landscapes of graph cliques encoded by node
sequences. The figure caption contains details on the instances and
relevant definitions can be found in section~\ref{sec:seq_clique}. We plot
the difference with the minimum cost $\tilde F(y)$, so that a plotted value
of $0$ means the global optimum has been found.

Our tentative conclusions are that the time to reach the optimal solution
scales moderately with problem size. The standard deviation over
realizations (error bars in the plot) also indicates a moderate variation
of optimization time across these randomly generated instances.

\section{Discussion and Conclusions}

In this contribution we have shown that, in principle, it is possible to
construct a genotypic encoding for any given phenotypically encoded
combinatorial optimization problem with the property that the encoded
landscape has no strict local minima. The construction hinges on three
ingredients: a cover-encoding map $\varphi:Y\to 2^X$ that satisfies a few
additional conditions, a suitable adjacency relation on $Y$, and an oracle
function that (miraculously) returns the optimal cost value on the
restrictions of the original problem to the covering sets $\varphi(y)$. Of
course, if we had such an oracle function in practice, we would not need a
search heuristic in the first place.

Nevertheless, the concepts of oracle functions and cover-encoding maps
  are not just an empty excercise. We have seen that cover-encoding maps
  $\varphi$ give rise to practically useful encodings \emph{provided} there
  is a good deterministic heuristic for the restriction of the optimization
  problem to $\varphi(y)$. For the NPP, it turns out that the
  Karmarkar-Karp differencing algorithm \cite{Karmarkar:82,Boettcher:08}
  provides a very good approximation to the oracle function. The
  prepartition encoding proposed by \citet{Ruml:1996}, on the other hand,
  ensures that the landscape of the oracle function is of the desirable
  type that has no local minima. Together these two facts make the work of
  \citet{Ruml:1996} a showcase application of the theory developed here.

The numerical simulations of Section~\ref{sect:num} strongly suggest that
encodings with local-minima-free landscapes indeed admit efficient
optimization by local search based methods also for other optimization
  problems. Hence the theoretical results obtained here are of practical
relevance provided a sufficiently accurate approximation to the oracle
function can be computed. The precise meaning of the phrase 'sufficiently
accurate approximation' remains an open question for future research. We
suspect, however, that the main problem arises when the approximation
claims $\alpha_f(y')<\alpha(y)$, suggesting that a step from $y$ to $y'$ be
accepted, while $\tilde F(y')>\tilde F(y)$ holds, suggesting the step to
$y'$ should not be taken.

The construction of encodings for several well-known optimization problems
also highlights the connections between encodings and a natural notion of
coarse-graining for optimization problems. This also suggests a link to
renormalization group methods commonly used in statistical physics. While
it is clear that there is not a trivial correspondence, and that real-space
coarse-grainings are just a particular subclass of encodings, this
connection certainly deserves further study. The formalism laid out here at
least provides a promising starting point.

An important issue in biology is that fact that encodings as symbolised by
the genotype-phenotype map, are themselves subject to evolutionary changes
because the mechanisms of development evolve. It is well known that
features of the genotype-phenotype, such as robustness \cite{Wagner:05} and
accessibility \cite{Fontana:98,Ndifon:09} have a key influence on evolution
in the long term. Mathematical approaches that focus on the properties of
encodings thus may become a very useful component in formal theories of
evolvability and developmental evolution.

\begin{acknowledgements}
  KK acknowledges funding from MINECO through the Ram{\'o}n y Cajal program
  and through project SPASIMM, FIS2016-80067-P (AEI/FEDER, EU).  This
  project has received funding from the European Research Council (ERC)
  under the European Union's Horizon 2020 research and innovation programme
  (grant agreement N.\ 694925).
\end{acknowledgements}

\bibliographystyle{spbasic} 
\bibliography{enc2}

\end{document}